\documentclass[conference]{IEEEtran}
\makeatletter
\def\ps@headings{%
\def\@oddhead{\mbox{}\scriptsize\rightmark \hfil \thepage}%
\def\@evenhead{\scriptsize\thepage \hfil \leftmark\mbox{}}%
\def\@oddfoot{}%
\def\@evenfoot{}}
\makeatother 

\IEEEoverridecommandlockouts
\usepackage{bbm}
\usepackage{amsfonts}
\usepackage[dvips]{graphicx}
\usepackage{times}
\usepackage{cite}
\usepackage{amsmath}
\usepackage{array}
\usepackage{amssymb}

\usepackage{stfloats}
\usepackage{slashbox}
\usepackage{graphicx}
\usepackage{footnote}
\usepackage{booktabs}
\usepackage{array}
\usepackage{algorithmic}
\usepackage{algorithm}
\usepackage{subeqnarray}
\usepackage{cases}
\usepackage{threeparttable}
\usepackage{color}
\usepackage{hyperref}
\usepackage{epstopdf}
\usepackage{bm}
\usepackage{multirow}
\usepackage[labelformat=simple]{subcaption}
\usepackage{adjustbox}

\newtheorem{theorem}{Theorem}

\newtheorem{lemma}{Lemma}

\newtheorem{corollary}{Corollary}

\begin{document}
\title{Analysis and Optimization of Probabilistic Caching in Multi-Antenna Small-Cell Networks}
\author{\authorblockN{Xianzhe Xu and Meixia Tao}
\authorblockA{Dept. of Electronic Engineering, Shanghai Jiao Tong University, Shanghai, China\\
Emails: \{august.xxz, mxtao\}@sjtu.edu.cn
}
}
\maketitle
\vspace{-1.5cm}
\begin{abstract}
Previous works on cache-enabled small-cell networks (SCNs) with probabilistic caching often assume that each user is connected to the nearest small base station (SBS) among all that have cached its desired content. The user may, however, suffer strong interference from other SBSs which do not cache the desired content but are geographically closer. In this work, we investigate this issue by deploying multiple antennas at each SBS. We first propose a user-centric SBS clustering model where each user chooses its serving SBS only from a cluster of $K$ nearest SBSs with $K$ being a fixed cluster size. Two beamforming schemes are considered. One is coordinated beamforming, where each SBS uses zero-forcing (ZF) beamformer to null out the interference within the coordination cluster. The other is uncoordinated beamforming, where each SBS simply applies matched-filter (MF) beamformer. Using tools from stochastic geometry, we obtain tractable expressions for the successful transmission probability (STP) of a typical user for both cases in the high signal-to-noise ratio (SNR) region. Tight approximations in closed-form expressions are also obtained. We then formulate and solve the optimal probabilistic caching problem to maximize the STP. Numerical results reveal interesting insights on the choices of ZF and MF beamforming in multi-antenna cache-enabled SCNs.
\end{abstract}
\section{Introduction}
In recent decades, mobile data traffic has grown rapidly due to the proliferation of mobile devices, which imposes heavy pressure on the limited backhaul link in cellular networks. Caching popular contents in small base stations (SBSs) at off-peak time is a promising way to alleviate the backhaul pressure by avoiding repetitive transmissions in the backhaul link. It can also improve user-perceived experience due to the reduced transmission latency.

Many prior works have studied caching strategies in cache-enabled small-cell networks (SCNs) under different optimization objectives. Thus far, there are two commonly used caching strategies. One is the most popular caching (MPC), where each SBS only caches the most popular contents until its cache size is full \cite{Bastug2015,yang2016analysis}. This strategy is suitable for networks where each user can only be associated with its nearest SBS.
The other is the probabilistic caching, where each SBS caches contents randomly with probabilities as design variables \cite{blaszczyszyn2015optimal,chenprobabilistic,chae2016caching,li2016optimization,7544526}. In particular, the work \cite{blaszczyszyn2015optimal} formulates and solves a probabilistic cache placement problem to maximize the cache hit rate using stochastic geometry. It shows that caching contents randomly with the optimized probabilities has significant gain over caching the most popular contents when each user can be covered by multiple SBSs. However, when a user is not served by the nearest SBS, the strong interference from the closer SBSs will severely harm the received signal-to-interference ratio (SIR). It is thus crucial to study interference management schemes with probabilistic caching in SCNs.

Some works have studied the transmission schemes to solve the above interference problem in cache-enabled networks. The authors in \cite{chen2016cooperative} make an initial attempt to apply joint transmission and successive interference cancelation (SIC) in cache-enabled SCNs and study the tradeoff between transmission diversity and content diversity. The authors in \cite{liu2017caching} employ zero-forcing (ZF) beamforming to null interference and optimize the cache placement for maximizing the success probability and the area spectral efficiency.

In this paper, we aim to exploit the benefits of multiple antennas in cache-enabled SCNs when probabilistic caching is adopted. Multiple antennas can be used to cancel interference from neighboring cells with SBS coordination, or only to strengthen the effective channel gain of desired signals without SBS coordination, both yielding significant increase in the received SIR. The work in \cite{chen2016cooperative} is limited to disjoint SBS clustering and the analysis only focuses on a user located at the cluster center. In this work, we consider a typical user and employ user-centric dynamic SBS clustering in cache-enabled SCNs. Each user is only allowed to be associated with a certain number of nearest SBSs and all SBSs cache files independently with identical probabilities. We investigate both ZF and matched-filter (MF) beamforming at each multi-antenna SBS. The main contributions are summarized as follows.

$\bullet$ We propose a user-centric dynamic SBS clustering model in multi-antenna cache-enabled SCNs, where each user is allowed to access one of the SBSs within its cluster. Each SBS adopts ZF-based coordinated beamforming when transmission coordination is allowed within the SBS cluster or MF beamforming otherwise.

$\bullet$ We obtain the closed-form (approximate) upper and lower bounds of the coverage probability of a typical user in both MF and ZF beamforming cases in the high signal-to-noise ratio (SNR) region using tools from stochastic geometry. Our analysis extends that in \cite{lee2015spectral} by allowing each user to be associated with its $k$-th nearest SBS, with $k$ being any integer not larger than the cluster size.

$\bullet$ We formulate an optimal probabilistic caching (OPC) problem to maximize the probability that the typical user successfully receives its requested file locally from the SBSs within its cluster. The problem is shown to be convex and solved by Lagrangian method with closed-form solutions.

$\bullet$ Numerical results demonstrate the gain of OPC over MPC. Results also reveal that when the number of antennas at each SBS is the same as the SBS cluster size, MF outperforms ZF. While when the number of antennas is larger than the cluster size, ZF performs better.

\emph{Notations}: This paper uses bold-face lower-case $\mathbf{h}$ for vectors and bold-face uppercase $\mathbf{H}$ for matrices. $\mathbf{H}^H$ is the conjugate transpose of $\mathbf{H}$ and $\mathbf{H}^\dagger$ is the left pseudo-inverse of $\mathbf{H}$, defined as $\mathbf{H}^\dagger=(\mathbf{H}^H\mathbf{H})^{-1}\mathbf{H}^H$. We use $\mathbf{I}_m$ to denote an $m\times m$ identity matrix and $\mathbf{0}_{1\times m}$ denotes a $1\times m$ zero vector.
\section{System Model}
We consider a multi-antenna cache-enabled SCN, where the locations of SBSs are modeled as a homogenous Poisson point process (HPPP)
$\Phi_b=\{\textbf{d}_i\in\mathbb{R}^2,\forall i\in\mathbb{N^+}\}$ with intensity $\lambda_b$. Each SBS is equipped with $L$ antennas. The locations of users are also modeled as a HPPP with intensity $\lambda_u$, which is independent with $\Phi_b$. Each user is equipped with a single antenna. We assume that $\lambda_u \gg \lambda_b$ so that the network is fully loaded with each SBS serving one user at a time.

We consider a file library $\mathcal{F}=\{f_1,f_2,\cdots,f_N\}$, where $N$ is the total number of files. All files are assumed to have the same normalized size of $1$. The popularity of file $f_n$ is $p_n$, satisfying $0\leq p_n\leq 1$ and $\sum_{n=1}^Np_n=1$. Without loss of generality, we assume $p_1\geq p_2\ldots\geq p_N$. Each SBS has a local cache that can store up to $M$ files with $M<N$. We adopt the \emph{probabilistic caching} strategy, where each SBS caches file $f_n$ with probability $b_n$ independently. Due to the cache size constraint and probability property, we have the constraints: $\sum_{n=1}^N b_n\leq M$ and $0\leq b_n \leq 1$ for $n=1,2,\ldots,N$.
\subsection{User Association Strategy}
In this work, each user is allowed to choose its serving SBS from a cluster of $K$ SBSs that are closest to the user, where $K\geq2$ is a positive integer. We shall refer to the $K$ closest SBSs of each user as the user-centric SBS cluster with size of $K$. When a user submits a file request, the nearest SBS within the cluster that has cached the requested file will serve the user. If none of the $K$ SBSs in the cluster caches the requested file, a macro base station (MBS) will download the file from the core network via backhaul link and then transmit it to the user. By such user-centric SBS clustering, the plane is tessellated into $K$-th order Voronoi cells, denoted as $\mathcal{V}_K(\textbf{d}_1,\cdots,\textbf{d}_K)$. The $K$-th order Voronoi cell associated with a set of $K$ points $\{\textbf{d}_1,\cdots,\textbf{d}_K\}$  is the region that all the points in this region are closer to these $K$ points than to any other point of $\Phi_b$, i.e., $\mathcal{V}_K(\textbf{d}_1,\cdots,\textbf{d}_K)=\{\textbf{d}\in\mathbb{R}^2|\cap_{k=1}^K \{\Vert \textbf{d}-\textbf{d}_k\Vert \leq \Vert \textbf{d}-\textbf{d}_i\Vert\},\textbf{d}_i\in \Phi_b\backslash\{\textbf{d}_1,\textbf{d}_2,\cdots,\textbf{d}_K\}\}$.

Without loss of generality, we focus on a typical user $u_0$ which is located at the origin, and can choose to connect to any of the SBSs in the $K$-th order Voronoi cell that it belongs to, denoted as $\mathcal{C}=\{\textbf{d}_1,\textbf{d}_2,\cdots,\textbf{d}_K\}$. The distance between $u_0$ and the $k$-th nearest SBS $\textbf{d}_k$ is $r_k$. As assumed earlier, the user intensity is much larger than SBS intensity, and hence the network is fully loaded with all the SBSs being active. We consider an interference-limited network where the noise can be neglected. Hence, the received signal of $u_0$ when associated with the $k$-th nearest SBS $\textbf{d}_k$, for $k=1,2,\ldots,K$ is given by:
\begin{equation}
y_0=r_k^{-\frac{\alpha}{2}}\mathbf{h}_{k0}\mathbf{w}_{k}x_{k}+\sum_{j\in\Phi_b\backslash
\{\textbf{d}_k\}}r_j^{-\frac{\alpha}{2}}\mathbf{h}_{j0}\mathbf{w}_{j}x_{j}, \label{eqn:receive}
\end{equation}
where the channel between $u_0$ and its $i$-th nearest SBS $\textbf{d}_i$ is assumed to have both small-scale fading, denoted as $\mathbf{h}_{i0}\in\mathbb{C}^{1 \times L}$, and large-scale fading $r_i^{-\frac{\alpha}{2}}$. The small-scale fading is modeled as Rayleigh fading, i.e., $\mathbf{h}_{i0}\sim \mathcal{CN}(\mathbf{0}_{1\times{L}},\mathbf{I}_{L})$ and the large-scale fading $r_i^{-\frac{\alpha}{2}}$ follows the distance-dependent power law model with $\alpha>2$ being the path loss exponent. $x_i$ and $\mathbf{w}_i\in \mathbb{C}^{L \times 1}$  denote the transmit signal and the associated beamformer vector, respectively.

We consider two types of beamforming design at each SBS. One is uncoordinated, where each SBS applies an MF based beamforming independently to maximize the effective channel gain of its own user. The other is coordinated, where the $K$ SBSs in each $K$-th order Voronoi cell apply ZF beamforming coordinately so that the intra-cluster interference can be nulled out completely. This requires that the number of antennas at each SBS should not be smaller than the cluster size, i.e. $L\geq K$. Note that since we consider a user-centric SBS clustering for coordinated beamforming, all the other $K-1$ users involved in the coordination must also locate in the same $K$-th order Voronoi cell as the typical user. This requirement is enforced to justify the \emph{typicality} of the typical user in our analysis and it can be easily satisfied given that the user intensity is much larger than the SBS intensity. This requirement is not needed for the uncoordinated beamforming though.
\subsection{Matched-Filter Beamforming}
When the typical user is associated with SBS $\textbf{d}_k\in\mathcal{C}$, the MF beamforming vector $\mathbf{w}_{k}$ at $\textbf{d}_k$ is given by:
\begin{equation}
\mathbf{w}_{k,\text{mf}}=\frac{\mathbf{h}_{k0}^H}{\Vert\mathbf{h}_{k0}\Vert}.
\end{equation}

Since each SBS serves its own user independently, the interference for the typical user comes from all SBSs except the serving SBS $\textbf{d}_k$ in the network. Thus, based on (\ref{eqn:receive}), the SIR for $u_0$ is given by:
\begin{equation}
\text{SIR}_{k,\text{mf}}=\frac{g_{k,\text{mf}}\cdot r_k^{-\alpha}}{\sum_{j\in\Phi_b \backslash
\{\textbf{d}_k\}} g_{j,\text{mf}}\cdot r_j^{-\alpha}},
\end{equation}
where $g_{k,\text{mf}}=\Vert\mathbf{h}_{k0}\Vert^2$ is the effective channel gain of the desired signal and follows the Gamma distribution with shape parameter $L$ and scale parameter $1$, denoted as $g_{k,\text{mf}}\sim \Gamma(L,1)$, and $g_{j,\text{mf}}=|\mathbf{h}_{j0}\mathbf{w}_{j,\text{mf}}|^2$ is the effective channel gain of the undesired signal and follows the exponential distribution with parameter $1$, denoted as $g_{j,\text{mf}}\sim \exp(1)$ \cite{jindal2011multi}.
\subsection{Zero-Forcing Beamforming}
In the ZF beamforming vector design, all the SBSs $\mathcal{C}=\{\textbf{d}_1,\textbf{d}_2,\cdots,\textbf{d}_K\}$ in the $K$-th order Voronoi cell $\mathcal{V}_K(\textbf{d}_1,\cdots,\textbf{d}_K)$ that the typical user $u_0$ falls into can coordinate to serve $K$ users (including $u_0$) located in the same $\mathcal{V}_K(\textbf{d}_1,\cdots,\textbf{d}_K)$ without intra-cluster interference \cite{lee2015spectral}. To ensure the feasibility, we assume $L \geq K$ as mentioned earlier. When $u_0$ is associated with SBS $\textbf{d}_k$, the beamforming vector at $\textbf{d}_k$ is given by:
\begin{equation}
\mathbf{w}_{k,\text{zf}}=\frac{\mathbf{(I}_{L}-\mathbf{H}\mathbf{H}^\dagger)\mathbf {h}_{k0}^T}{\Vert\mathbf{(I}_{L}-\mathbf{H}\mathbf{H}^\dagger)\mathbf{h}_{k0}^T\Vert},\label{eqn:zf}
\end{equation}
where $\mathbf{H}=[\mathbf{h}_{k1}^T,\mathbf{h}_{k2}^T,\cdots,\mathbf{h}_{k(K-1)}^T]$ is the channel between SBS $\textbf{d}_k$ and the other $K-1$ users in the coordination group. By (\ref{eqn:zf}), the interference caused by SBS $\textbf{d}_k$ to all other $K-1$ users in the cluster is canceled. The other $K-1$ SBSs from $\mathcal{C}$ adopt the same ZF beamforming strategy. Thus, the interference for the typical user only comes from the SBSs out of the cluster. Therefore, the SIR for the typical user is given by:
\begin{equation}
\text{SIR}_{k,\text{zf}}=\frac{g_{k,\text{zf}}\cdot r_k^{-\alpha}}{\sum_{j\in\Phi_b\backslash \mathcal{C}
}g_{j,\text{zf}}\cdot r_j^{-\alpha}},
\end{equation}
where $g_{k,\text{zf}}=|\mathbf{h}_{k0}\mathbf{w}_{k,\text{zf}}|^2$ is the effective channel gain of the desired signal and follows $g_{k,\text{zf}}\sim \Gamma(L-K+1,1)$, and $g_{j,\text{zf}}=|\mathbf{h}_{j0}\mathbf{w}_{j,\text{zf}}|^2$ is the effective channel gain of the undesired signal and follows $g_{j,\text{zf}}\sim \exp(1)$ \cite{lee2015spectral,jindal2011multi,li2015user}. Note again that the ZF-based coordination is applied to the SBS set $\mathcal{C}$ to serve $K$ users within the same $\mathcal{V}_K(\textbf{d}_1,\cdots,\textbf{d}_K)$ (to justify the typicality of $u_0$). It does not apply to SBSs across different $K$-th order Voronoi cells.

Before concluding this section, we would like to remark that ZF beamforming requires each SBS to know the channel state information (CSI) of all the $K$ users in the cluster, therefore extra signalling overhead is needed. On the other hand, MF beamforming only requires each SBS to know the CSI of its own user.
\section{PERFORMANCE ANALYSIS}
In this paper, we use \emph{successful transmission probability} (STP) of the typical user as the performance metric.

The STP, denoted as $P_{\text{suc}}(K)$, is defined as the probability that the typical user can successfully receive its requested file locally from the cluster of $K$ closest SBSs without resorting to the core network. A file can be successfully and locally transmitted if and only if it is cached and the received SIR exceeds a given SIR target. The probability that the received SIR exceeds a SIR target is called \emph{coverage probability}. Denote $P_{\text{cov}}^{k}(K)$ as the coverage probability of the typical user served by the $k$-th $(k=1,2\ldots K)$ nearest SBS in the cluster and it is given by:
\begin{equation}
P_{\text{cov}}^{k}(K)=P[\text{SIR}_k \geq \gamma],
\end{equation}
where $\gamma$ is the given SIR target.

Therefore, the STP is given by:
\begin{equation}
P_{\text{suc}}(K)=\sum_{n=1}^N p_n{\sum_{k=1}^K b_n(1-b_n)^{k-1}P_{\text{cov}}^{k}(K)}, \label{eqn:suc}
\end{equation}
where $p_n$ is the request probability of file $f_n$ and $b_n$ is the cache probability of file $f_n$.

Our goal is to analyze $P_{\text{suc}}(K)$ with MF and ZF beamforming, respectively, then optimize the cache probabilities for both schemes and finally compare their performances.

In the rest of this section, we analyze the coverage probability $P_{\text{cov}}^{k}(K)$ in both MF and ZF beamforming schemes. Note that the works \cite{lee2015spectral,
li2015user,di2015stochastic} have studied the SIR distribution in multi-antenna networks with coordinated beamforming using stochastic geometry tools. In particular, \cite{di2015stochastic} obtains the characteristic function of other-cell interference when both base stations and users are equipped with multiple antennas. The work \cite{li2015user} obtains a tractable and approximate expression of the complementary cumulative distribution function (CCDF) of SIR when the user is served by the nearest SBS only. The work \cite{lee2015spectral} further obtains closed-form upper and lower bounds of the CCDF of SIR. Our analysis differs from the above references in that each user is allowed to be served by the $k$-th nearest SBS, for $k=1,2,\ldots,K$, which makes our problem more challenging and more general.
\subsection{Coverage Probability with MF Beamforming}
\begin{lemma} \label{lemma:1}
The coverage probability of the typical user served by the $k$-th nearest SBS with MF beamforming is:
\begin{align}
P_{\text{cov,mf}}^{k}(K)=\mathbb{E}_{r_k}\left[\sum_{i=0}^{L-1}\frac{(-\gamma {r_k}^\alpha)^i}{i!}\mathcal{L}_{I_r}^{(i)}(\gamma {r_k}^\alpha)\right],\label{eqn:coverage}
\end{align}
where $I_{r}=\sum_{j\in\Phi_b\backslash\{\textbf{d}_k\}}g_{j,\text{mf}}\cdot r_j^{-\alpha}$, $\mathcal{L}_{I_r}(s)=\mathbb {E}[e^{-sI_r}]$ is the Laplace transform of $I_{r}$, which is given by:
\begin{align}
\mathcal{L}_{I_r}(s)&=\left(\int_0^{r_k} \frac{1}{1+sr^{-\alpha}}\frac{2r}{r_k^2}dr\right)^{k-1} \nonumber\\ &~~~\times\exp\left(-2\pi\lambda_b\int_{r_k}^{\infty}\frac{sr^{-\alpha}}{1+sr^{-\alpha}}rdr\right),
\end{align}
and $\mathcal{L}_{I_r}^{(i)}(s)$ is the $i$-th order derivative of $\mathcal{L}_{I_r}(s)$.
\end{lemma}

\begin{proof}
See Appendix A.
\end{proof}

With the expression (\ref{eqn:coverage}), we still need to obtain the expectation over $r_k$. The probability density function (pdf) of $r_k$ is given in Lemma 2.
\begin{lemma}[\cite{haenggi2005distances}]
Given a HPPP in the 2-dimensional plane with intensity $\lambda_b$, the distance $R_k$ of the $k$-th nearest point to the origin follows the generalized Gamma distribution:
\begin{equation}
f_{R_k}(r)=\frac{2\left(\lambda_b\pi r^2\right)^k}{r\Gamma(k)}\exp\left(-\lambda_b \pi r^2\right).
\end{equation}
\end{lemma}

With Lemma 1 and Lemma 2, we can get a tractable expression of the coverage probability. In the following, we provide more compact forms to bound the coverage probability.

\begin{theorem} \label{theorem:MF bound}
The coverage probability of the typical user served by the $k$-th nearest SBS using MF beamforming is bounded as:
\begin{align}
P_{\text{cov,mf}}^{k,\text{l}}(K) \leq P_{\text{cov,mf}}^{k}(K) \leq P_{\text{cov,mf}}^{k,\text{u}}(K),
\end{align}
with
\begin{align}
&P_{\text{cov,mf}}^{k,\text{u}}(K)= \sum_{l=1}^{L} \beta_1\left(\eta,\gamma,\alpha,l,k\right)\frac{\binom{L}{l}(-1)^{l+1}}{\left(1+\beta_2\left(\eta,\gamma,\alpha,l\right)\right)^k},\label{eqn:app1} \\
&P_{\text{cov,mf}}^{k,\text{l}}(K)=\sum_{l=1}^{L} \beta_1(1,\gamma,\alpha,l,k)\frac{\binom{L}{l}(-1)^{l+1}}{(1+\beta_2(1,\gamma,\alpha,l))^k} \label{eqn:app11},
\end{align}
where
\begin{align}
&\beta_1(x,\gamma,\alpha,l,k)=\left[1-\frac{2 (x \gamma l)^{\frac{2}{\alpha}}}{ \alpha}B\left(\frac{2}{\alpha},1-\frac{2}{\alpha},\frac{1}{1+x \gamma l}\right)\right]^{k-1},\label{eqn:beta1} \\
&\beta_2(x,\gamma,\alpha,l)=2\frac{\left(x \gamma l\right)^{\frac{2}{\alpha}}}{\alpha} B^{'}\left(\frac{2}{\alpha},1-\frac{2}{\alpha},\frac{1}{1+x \gamma l}\right)\label{eqn:beta2},
\end{align}
and $\eta=(L!)^{-\frac{1}{L}}$, $B(x,y,z)\triangleq \int_0^z u^{x-1}(1-u)^{y-1}du$ is the incomplete Beta function and $B^{'}(x,y,z)\triangleq \int_z^1 u^{x-1}(1-u)^{y-1}du$ is the complementary incomplete Beta function.
\end{theorem}

\begin{proof}
See Appendix B.
\end{proof}

In the special cases with $L=1$ (single antenna), the upper and lower bounds coincide and hence give the exact expression. In addition, if $\alpha=4$, this expression can be written as a closed form.

\begin{corollary}
The coverage probability of the typical user served by the $k$-th nearest SBS in the single-antenna network with $\alpha=4$ is given by:
\begin{align}
&P_{\text{cov,mf}}^k(K)=\frac{\left(1-\sqrt{\gamma}\text{arcsin}\frac{1}{\sqrt{1+\gamma}}\right)^{k-1}}{\left(1+\sqrt{\gamma}\text{arccos}\frac{1}{\sqrt{1+\gamma}}\right)^k}.\label{eqn:special1}
\end{align}
\end{corollary}

In the special case, when the user is served by its nearest SBS, e.g., $k=1$, (\ref{eqn:special1}) reduces to \cite[Theorem $2$]{andrews2011tractable}.

\subsection{Coverage Probability with ZF Beamforming}
\begin{lemma}
The coverage probability of the typical user served by the $k$-th nearest SBS with ZF beamforming is:
\begin{align}
P_{\text{cov,zf}}^{k}(K)=\mathbb{E}_{r_k,r_K}\left[\sum_{i=0}^{L-K}\frac{(-\gamma {r_k}^\alpha)^i}{i!}\mathcal{L}_{I_r}^{(i)}(\gamma {r_k}^\alpha)\right], \label{eqn:cov-zf}
\end{align}
where $I_r=\sum_{j\in\Phi_b\backslash \mathcal{C}}g_{j,\text{zf}}\cdot r_j^{-\alpha}$ and its Laplace transform is given by:
\begin{align}
\mathcal{L}_{I_r}(s)=\exp\left(-2\pi\lambda_b\int_{r_K}^{\infty}\frac{sr^{-\alpha}}{1+sr^{-\alpha}}rdr\right).
\end{align}
\end{lemma}

The proof of lemma $3$ is similar to Appendix A, so we omit it here. Notice that the expectation in (\ref{eqn:cov-zf}) is not only over $r_k$, but also over $r_K$. Thus, we need to know the joint pdf of $r_k$ and $r_K$. The pdf of $r_K$ is given in Lemma $2$, we next focus on the pdf of $r_k$ conditioned on $r_K$.

\begin{lemma}[\cite{srinivasa2010distance}]
Consider $K$ points randomly located in a $2$-dimensional circle of radius $r_K$ centered at the origin according to the uniform distribution, then the distance $R_k$ from the origin to the $k$-th nearest point follows the distribution given by:
\begin{align}
f_{R_k|R_K}(r_k|r_K)=\frac{2r_k}{r_K^2 B(K-k,k)}\left(\frac{r_k^2}{r_K^2}\right)^{k-1}\left(1-\frac{r_k^2}{r_K^2}\right)^{K-k-1},
\end{align}
where $B(x,y)$ denotes the Beta function given by $B(x,y)=\int_0^1t^{x-1}(1-t)^{y-1}dt$.
\end{lemma}

Thus, the joint pdf of $r_k$ and $r_K$ can be obtained as:
\begin{align}
f_{R_k,R_K}(r_k,r_K)&=f_{R_k|R_K}(r_k|r_K)f_{R_K}(r_K) \nonumber \\
&=\frac{4(\lambda_b\pi)^K}{\Gamma(K-k)\Gamma(k)}r_kr_K(r_k^2)^{k-1} \nonumber \\
&~~~\times(r_K^2-r_k^2)^{K-k-1}\exp\left(-\lambda_b \pi r_K^2\right). \label{eqn:joint pdf}
\end{align}

With Lemma 3 and joint pdf of $r_k$ and $r_K$, we can obtain a tractable expression for the coverage probability. More compact forms of the approximate coverage probability bounds are obtained in the following theorem.
\begin{theorem}\label{theorem:BF bound}
The coverage probability of the typical user served by the $k$-th nearest SBS with ZF beamforming can be approximately bounded as:
\begin{align}
P_{\text{cov,zf}}^{k,\text{l}}(K) \lesssim P_{\text{cov,zf}}^{k}(K) \lesssim P_{\text{cov,zf}}^{k,\text{u}}(K),
\end{align}
with
\begin{align}
&P_{\text{cov,zf}}^{k,\text{u}}(K)=\sum_{l=1}^{L-K+1} \frac{\binom{L-K+1}{l}(-1)^{l+1}}{{\left[1+\left(\kappa \gamma l\right)^{\frac{2}{\alpha}}\sqrt{\frac{k}{K}} \mathcal{A}\left(\frac{\sqrt K \left(\kappa \gamma l\right)^{-\frac{2}{\alpha}}}{\sqrt k }\right)\right]^k}},\label{eqn:app2}  \\
&P_{\text{cov,zf}}^{k,\text{l}}(K)=\sum_{l=1}^{L-K+1} \frac{\binom{L-K+1}{l}(-1)^{l+1}}{{\left[1+\left( \gamma l\right)^{\frac{2}{\alpha}}\sqrt{\frac{k}{K}} \mathcal{A}\left(\frac{\sqrt K \left(\gamma l\right)^{-\frac{2}{\alpha}}}{\sqrt k }\right)\right]^k}}, \label{eqn:app22} \end{align}
where $\mathcal{A}(x)=\int_x^\infty \frac{1}{1+u^{\frac{\alpha}{2}}}du$ and $\kappa=(L-K+1!)^{-\frac{1}{L-K+1}}$.
\end{theorem}

\begin{proof}
See Appendix C.
\end{proof}

The approximate upper and lower bounds coincide when $L=K$. Furthermore, when $\alpha=4$, the approximate coverage probability can be written in a closed form.

\begin{corollary}
The approximate coverage probability of the typical user served by the $k$-th nearest SBS using ZF beamforming with $L=K$ and $\alpha=4$ is given by:
\begin{align}
P_{\text{cov,zf}}^k(K)\simeq \frac{1}{\left[1+\sqrt{\frac{k\gamma}{K}}\text{arccot}\left(\frac{K}{k\gamma}\right)\right]^k}.\label{eqn:special2}
\end{align}
\end{corollary}

When the user is served by the nearest SBS, e.g., $k=1$, (\ref{eqn:special2}) reduces to \cite[Eqn. (28)]{lee2015spectral}.

As we shall demonstrate numerically in Section \uppercase\expandafter{\romannumeral5}-A, the (approximate) upper bounds (\ref{eqn:app1}) and (\ref{eqn:app2}) are very tight. Therefore, we shall use them as approximate coverage probabilities for caching optimization.
\begin{figure*}[htb]
\begin{minipage}[t]{0.5\linewidth}
\centering
\includegraphics[width=6.75cm]{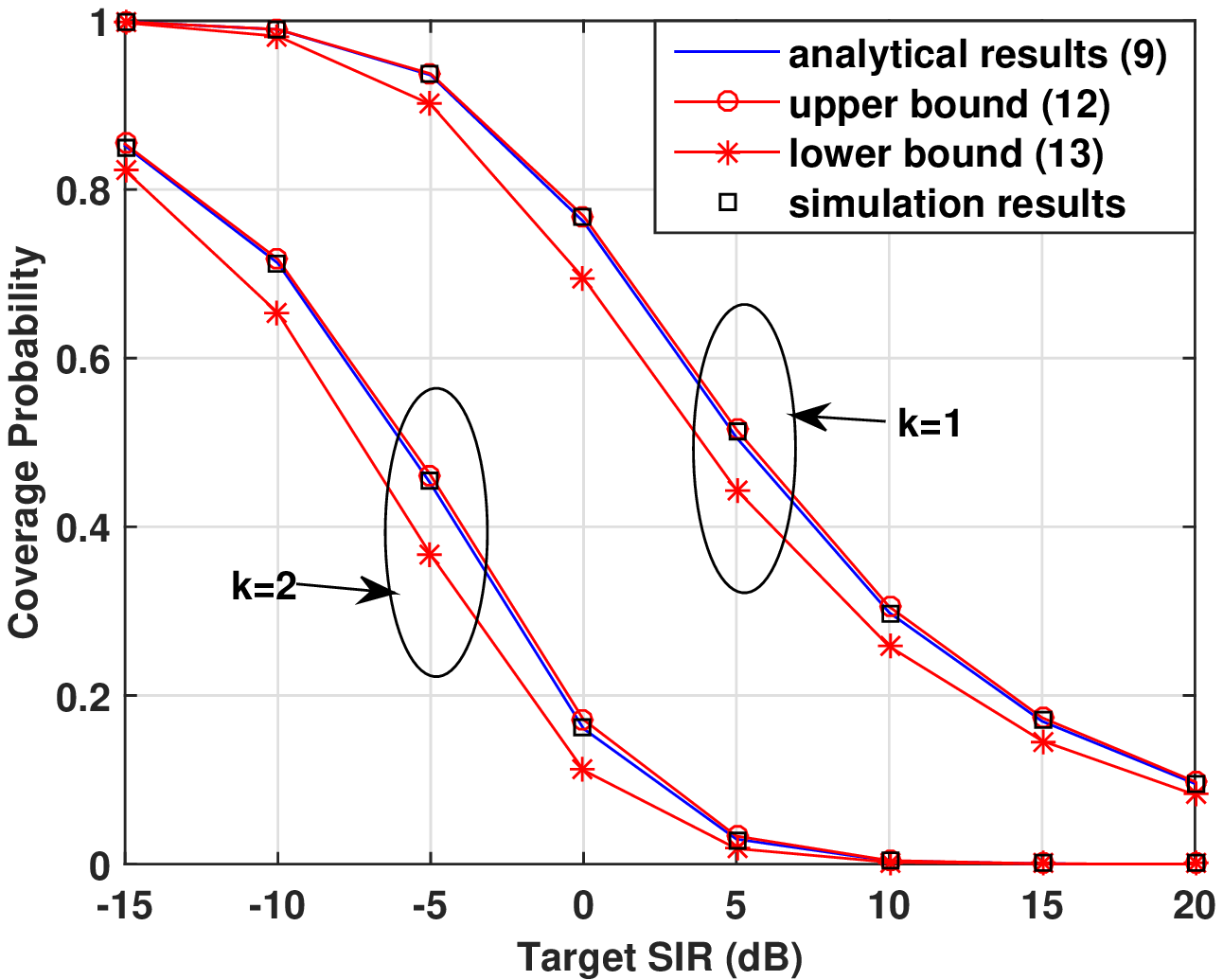}\\
\caption{Coverage probability with $K=2$ and $L=2$ in MF beamforming.}\label{fig:MF}
\end{minipage}
\hfill
\begin{minipage}[t]{0.5\linewidth}
\centering
\includegraphics[width=6.75cm]{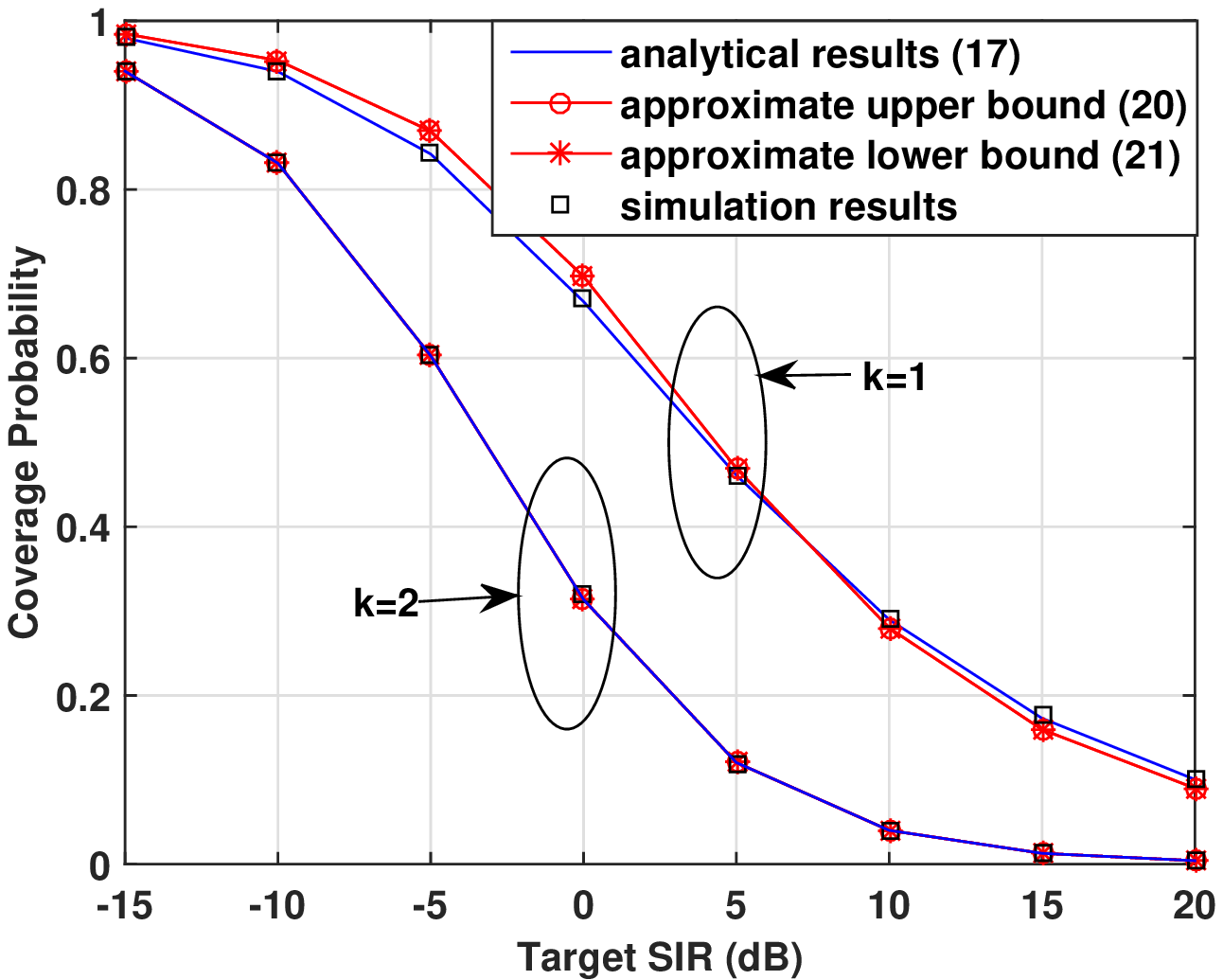}\\
\caption{Coverage probability with $K=2$ and $L=2$ in ZF beamforming.}\label{fig:BF}
\end{minipage}
\end{figure*}

\begin{figure*}[htb]
\begin{minipage}[t]{0.5\linewidth}
\centering
\includegraphics[width=6.75cm]{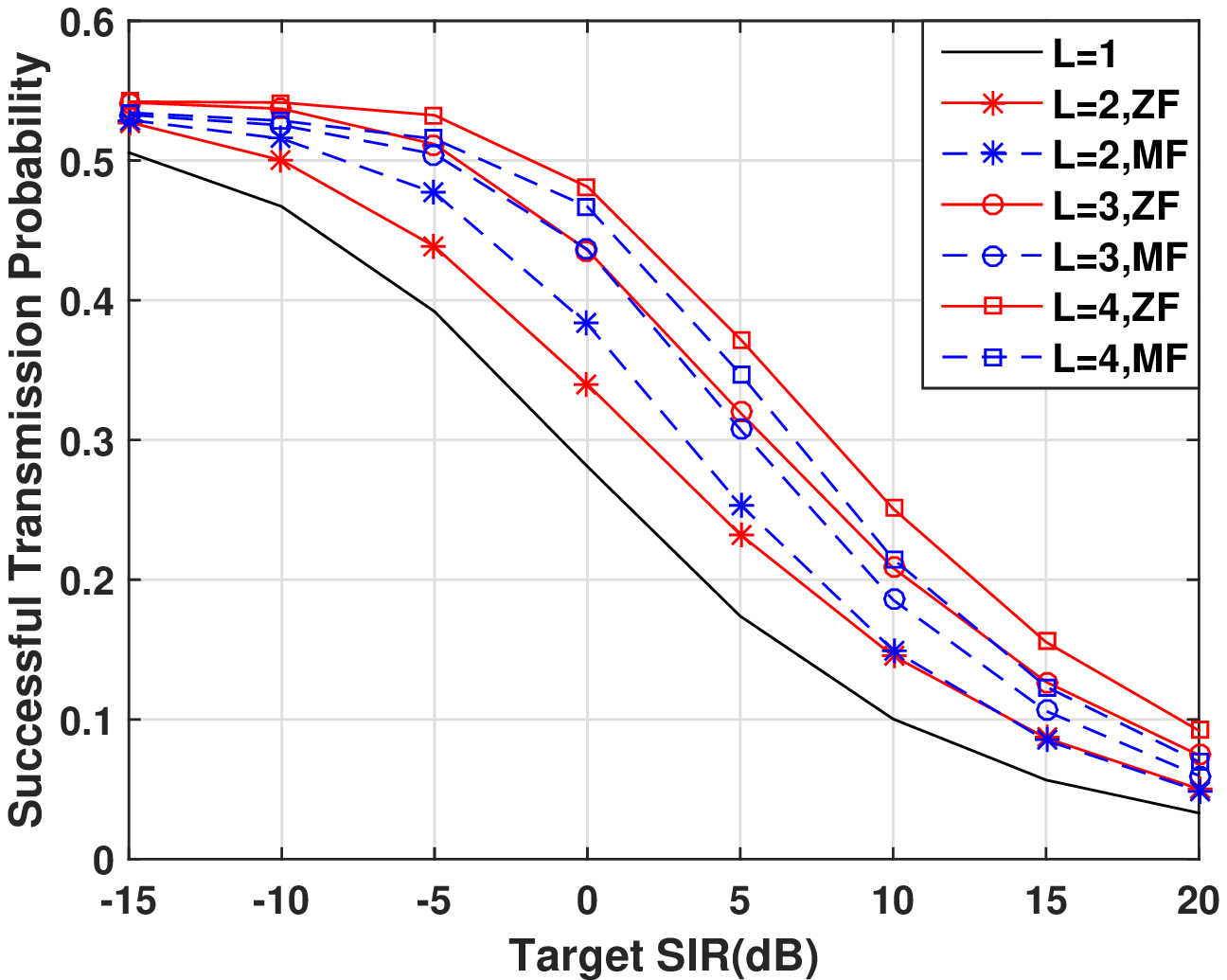}\\
\caption{MF vs ZF with different number of antennas with $K=2$.} \label{fig:Antennas}
\end{minipage}
\hfill
\begin{minipage}[t]{0.5\linewidth}
\centering
\includegraphics[width=6.75cm]{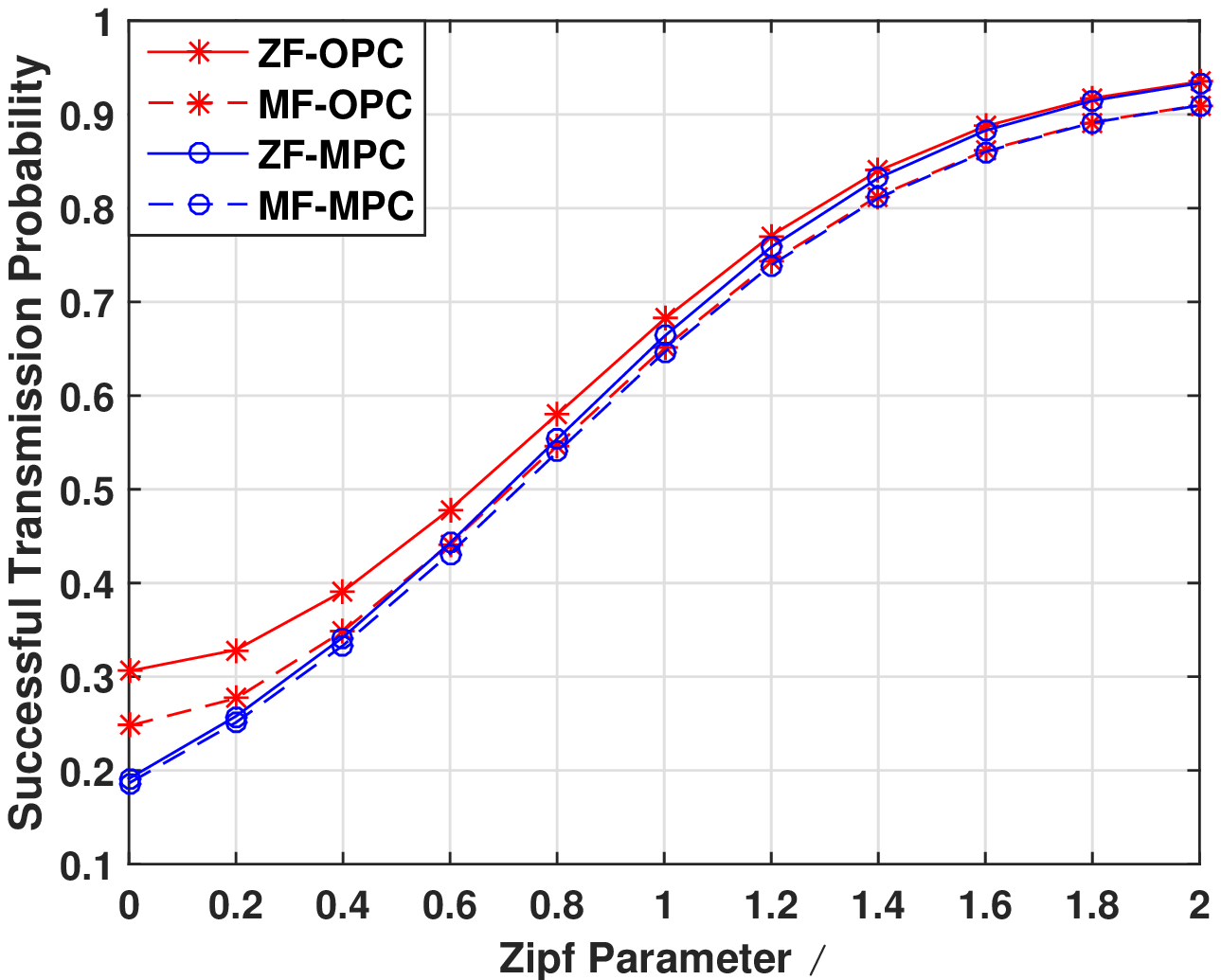}\\
\caption{STP for Different Zipf Skewness Parameter $\delta$ with $K=2$, $L=4$ and $M=20$.}\label{fig:delta}
\end{minipage}
\end{figure*}


\section{Caching Optimization}
In this section, we find the optimal cache probabilities to maximize the STP. The corresponding problem can be formulated as:
\addtocounter{equation}{1}
\begin{align}
\text{\textbf{P1:}} \quad \underset{\{b_n\}}{\text{max}} \quad &\sum_{n=1}^N p_n\sum_{k=1}^K b_n(1-b_n)^{k-1}P^k_{\text{cov}}(K), \tag{\theequation a}  \label{eqn:objective}\\
\text{s.t} \quad &\sum_{n=1}^N b_n \leq M,   \tag{\theequation b}   \label{eqn:contraint1}  \\
 &0\leq b_n \leq 1, ~~~  n=1,2,\ldots,N,     \tag{\theequation c}   \label{eqn:contraint2}
\end{align}
where $P_{\text{cov}}^k(K)$ is given by (\ref{eqn:app1}) for MF  beamforming and by (\ref{eqn:app2}) for ZF beamforming. (\ref{eqn:contraint1}) is the cache size constraint. A practical cache placement realization with given probabilities can be found in \cite{blaszczyszyn2015optimal}. Since caching more files increases the STP, without loss of optimality, constraint (\ref{eqn:contraint1}) can be rewritten as:
\begin{align}
\sum_{n=1}^N b_n=M.   \label{eqn:constraint}
\end{align}

\begin{lemma}
The problem $\textbf{P1}$ is convex for both MF and ZF beamforming.
\end{lemma}

\begin{proof}
See Appendix D.
\end{proof}

By using KKT condition, the optimal solutions of $\textbf{P1}$ satisfy the conditions as follows.
\begin{theorem}
The optimal cache probabilities of $\textbf{P1}$ satisfy£º
\begin{align}
b_n(\mu^*)=\text{min}\left(1,[w_n(\mu^*)]^+\right),  \label{eqn:11111}
\end{align}
where $[z]^+=\text{max}(z,0)$, $\mu^*\geq0$ is the optimal dual variable satisfying the cache size constraint (\ref{eqn:constraint}) and $w_n(\mu^*)$ is the real root of the function:
\begin{align}
p_n\sum_{k=1}^K[1-w_n(\mu^*)]^{k-2}[1-kw_n(\mu^*)]P_{\text{cov}}^k(K)-\mu^*=0. \label{eqn:sum}
\end{align}
\end{theorem}

\begin{proof}
See Appendix E.
\end{proof}
To obtain the optimal cache strategy, we should find the optimal dual variable $\mu^*$ by substituting (\ref{eqn:11111}) into the cache size constraint $\sum_{n=1}^Nb_n(\mu^*)=M$. It is observed that $b_n(\mu)$ is a decreasing function of $\mu$. Thus, the sum of $b_n(\mu)$ is also decreasing. Therefore, we can use the bisection method to find the optimal $\mu^*$. The algorithm is given in Algorithm \ref{alg:1}.
\begin{algorithm}
\caption{A Bisection Method to Finding optimal cache strategy} \label{alg:1}
\begin{algorithmic}[1]
\STATE Initialize the upper and lower bound of $\mu$ as $\mu_u=p_1\sum_{k=1}^KP_{cov}^k(K)$ and $\mu_l=p_N\left[P_{cov}^1(K)-P_{cov}^2(K)\right]$
\WHILE{$\mu_u-\mu_l \geq \epsilon (\text{error tolerance level})$}
\STATE $\mu^*=(\mu_u+\mu_l)/2$
\IF{$\sum_{n=1}^Nb_n(\mu^*) < M $}
\STATE update $\mu_u=\mu^*$
\ELSE
\STATE update $\mu_l=\mu^*$
\ENDIF
\ENDWHILE
\FOR{$n=1$:$N$}
\STATE $b_n(\mu^*)=\text{min}(1,[w_n(\mu^*)]^+)$
\ENDFOR
\end{algorithmic}
\end{algorithm}
\section{Simulation Results and Discussion}
In this section, we first validate the tightness of the approximate coverage probability in both MF and ZF schemes by simulations. Next, we analyze the effects of number of antennas, and Zipf parameter on STP with OPC and compare the performance between ZF and MF. Besides, comparison with MPC is given in numerical results. Finally, the optimal cache probabilities are proposed in both MF and ZF schemes.

Simulations are performed in a square area of $4000\times4000$ $\text{m}^2$. The file popularity is modeled as the Zipf distribution, i.e., $p_n=\frac{1/n^\delta}{\sum_{j=1}^N1/j^\delta}$ for file $f_n$ with $\delta$ being Zipf skewness parameter. Unless otherwise stated, the simulation parameters are set as follows:
SBS intensity $\lambda_b=5\times10^{-5}/ \text{m}^2$, path loss exponent $\alpha=4$, number of total files $N=100$, cache size $M=10$, SIR target $\gamma=0$~dB and Zipf skewness parameter $\delta=0.9$.

\subsection{Validation of Analytical Results}
From Fig. \ref{fig:MF} and Fig. \ref{fig:BF}, we observe that the simulation and analytical results of the coverage probabilities match well for both MF and ZF beamforming. It is also seen that (\ref{eqn:app1}) and (\ref{eqn:app2}) are very tight. As such, we shall use them to approximate the coverage probabilities to find the optimal cache probabilities in the optimization problem $\mathbf {P1}$.
\subsection{Number of antennas}
Fig. \ref{fig:Antennas} illustrates the STP with the optimized caching probabilities for different number of antennas. It is seen that when the cluster size $K$ is fixed, increasing the number of antennas increases the STP for both MF and ZF schemes but the gain diminishes as $L$ grows. It is also seen that when $L>K$ (e.g. $L=3, 4$), ZF beamforming always outperforms MF beamforming at all the considered SIR targets. Otherwise when $L=K$, MF beamforming performs better than ZF beamforming at most SIR threshold (e.g. $\gamma<10$dB) and they are almost the same only at high SIR threshold  (e.g. $\gamma>10$dB). This is because when the number of antennas equals the cluster size, the effective channel gain of the desired signal with ZF beamforming is much smaller than that of MF beamforming although the former suffers less interference. However, when $L$ is larger than $K$, the SBS has enough spatial dimensions to null out the intra-cluster interference and strengthen the effective channel gain of the desired signals simultaneously. Therefore, ZF outperforms MF.
\begin{figure}[t]
\begin{centering}
\vspace{-0.2cm}
\includegraphics[scale=.6]{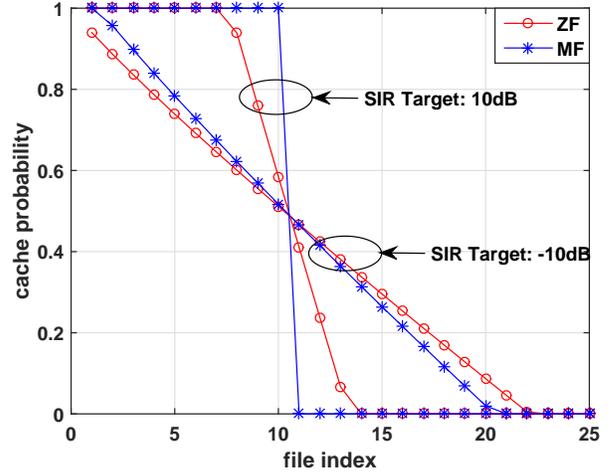}
\vspace{-0.1cm}
\caption{\small{Cache Probability with $K=2$ and $L=4$.}} \label{fig:cache_probability}
\end{centering}
\vspace{-0.3cm}
\end{figure}
\subsection{Zipf Parameter}
In Fig. \ref{fig:delta}, it is observed that STP grows when Zipf parameter $\delta$ becomes larger. This is because larger $\delta$ means few popular files are requested with a larger probability. Besides, the performance of MPC gets closer to that of OPC when $\delta$ grows. When $\delta > 1.8$, OPC performs almost identically to MPC. Moreover, it is observed that the gap between ZF and MF in OPC becomes smaller when $\delta$ increases. This is because the coverage probability of the user not served by the nearest SBS in ZF is larger than that in MF. For small $\delta$, the file popularity distribution is more like uniform distribution, which causes that SBSs prefer to cache more number of different files and users have a larger probability to connect to farther SBSs. In this case, ZF can make a better utilization of the multiple SBSs. While for large $\delta$, OPC degenerates to MPC, and thus users have a high probability to connect with the nearest SBS and the performances of ZF and MF become closer.
\subsection{Optimal Cache Strategy}
Fig. \ref{fig:cache_probability} shows that at the high SIR target with $\gamma=10$dB, the typical user most likely can only connect to the nearest SBS, and thus the most popular caching strategy is almost optimal. When the SIR target is small with $\gamma=-10$dB, the typical user can be served by farther SBSs. Therefore, caching more number of different files is better. The optimal cache solutions in MF are closer to MPC compared to ZF since the users not served by the nearest SBS in MF case suffers strong interference and thus users are more willing to connect to the nearest SBS.

\section{Conclusions}
In this work, we focus on the performance of ZF and MF beamforming in multi-antenna cache-enabled SCNs. Tractable and closed-form approximate expressions of the STP in both MF and ZF schemes are obtained by using tools from stochastic geometry. We then formulate the optimal probabilistic caching problem for maximizing the STP, which is proven to be a convex optimization. This problem is solved and the closed-form optimal cache solutions are obtained. We numerically analyze the effects of the number of antennas and Zipf parameter on the STP and make comparisons between ZF and MF beamforming. Numerical results also reveal that ZF outperforms MF when the number of antennas is larger than the cluster size and MF performs better when the number of antennas equals the cluster size.

\begin{appendices}
\section{Proof of Lemma $1$}
The coverage probability can be written as:
\begin{align}
\allowdisplaybreaks
P_{\text{cov,mf}}^{k}(K)&=P\left[\frac{g_{k,\text{mf}}\cdot r_k^{-\alpha}}{\sum_{j\in\Phi_b\backslash
\{\textbf{d}_k\}} g_{j,\text{mf}}\cdot r_j^{-\alpha}} \geq \gamma\right] \nonumber \\
&=\mathbb{E}_{r_k,I_r}\left[P\left[g_{k,\text{mf}} \geq \gamma r_k^{\alpha} I_r\right]|r_k,I_r\right] \nonumber \\
&\overset{(a)}{=}\mathbb{E}_{r_k,I_r}\left[\sum_{i=0}^{L-1}\frac{(\gamma r_k^{\alpha} I_r)^i}{i!}e^{-\gamma r_k^{\alpha} I_r}|r_k,I_r\right] \nonumber \\
&\overset{(b)}{=}\mathbb{E}_{r_k}\left[\sum_{i=0}^{L-1}\frac{(-\gamma {r_k}^\alpha)^i}{i!}\mathcal{L}_{I_r}^{(i)}(\gamma {r_k}^\alpha)|r_k\right],
\end{align}
where step (a) follows from the series expansion of the CCDF $\overline{F}(x;m,\theta)$ for Gamma distribution $\Gamma(m,\theta)$ when $\theta$ is a positive integer, i.e., $\overline{F}(x;m,\theta)=\sum_{i=0}^{m-1}\frac{1}{i!}(\frac{x}{\theta})^i e^{-\frac{x}{\theta}}$; and (b) follows from the derivative property of the Laplace transform: $\mathbb{E}[X^ie^{-sX}]=(-1)^i \mathcal{L}_X^{(i)}(s)$.

The interference $I_r$ consists of two parts, the interference $I_1$ from the $k-1$ SBSs closer to $u_0$ than the serving SBS $\textbf{d}_k$ and the interference $I_2$ from all the SBSs farther than $\textbf{d}_k$. Therefore, the Laplace transform of interference $\mathcal{L}_{I_r}(s)$ can be given by the product of $\mathcal{L}_{I_1}(s)$ and $\mathcal{L}_{I_2}(s)$.
The Laplace transform of $I_1$ is given by:
\begin{align}
\mathcal{L}_{I_1}(s)&=\mathbb{E}_{\Phi_b,g_{j,\text{mf}}}\left[\prod_{j\in\Phi_b\bigcap \mathcal{B}(0,r_k) \backslash \{\textbf{d}_k\}}\exp\left(-sg_{j,\text{mf}}\cdot r_j^{-\alpha}\right)\right] \nonumber \\
&\overset{(a)}{=}\mathbb{E}_{\Phi_b}\left[\prod_{j\in\Phi_b\bigcap \mathcal{B}(0,r_k) \backslash \{\textbf{d}_k\}}\frac{1}{1+sr_j^{-\alpha}}\right]  \nonumber \\
&\overset{(b)}{=}\left(\int_0^{r_k} \frac{1}{1+sr^{-\alpha}}\frac{2r}{r_k^2}dr\right)^{k-1}, \label{eqn:L1}
\end{align}
where step (a) follows from the i.i.d. distribution of $g_{j,\text{mf}}$ and its further independence from $\Phi_b$, and step (b) follows since these $k-1$ SBSs are i.i.d. and uniformly distributed in the circle $\mathcal{B}(0,r_k)$ and the pdf of the distance between these SBSs and the typical user $u_0$ is given by:
\begin{align}
f_{R}(r)=
\begin{cases}
\frac{2r}{r_k^2},&0\leq r \leq r_k \\
0,&r_k \leq r
\end{cases}.
\end{align}

The Laplace transform of $I_2$ is given by:
\begin{align}
\mathcal{L}_{I_2}(s)
&=\mathbb{E}_{\Phi_b}\left[\prod_{j\in\Phi_b\backslash \mathcal{B}(0,r_k)}\frac{1}{1+sr_j^{-\alpha}}\right]  \nonumber \\
&=\exp\left(-2\pi\lambda_b\int_{r_k}^{\infty}\frac{sr^{-\alpha}}{1+sr^{-\alpha}}rdr\right), \label{eqn:L2}
\end{align}
where the last step follows from the probability generating functional (PGFL) of the HPPP.
\section{Proof of Theorem $1$}
The proof relies on the following lemma.
\begin{lemma} [Alzer's Inequality \cite{huang2007space, alzer1997some}]

If $Z$ is a random variable following the Gamma distribution $Z\sim\Gamma(T,1)$, the cumulative distribution function (CDF) $F_Z(z)=P[Z \leq z]$ is bounded by:
\begin{align}
(1-e^{-a z})^T \leq F_Z(z) \leq (1-e^{-z})^T,
\end{align}
where $a=(T!)^{-\frac{1}{T}}$.
\end{lemma}

By Lemma 6, the CCDF of $Z$ can be upper bounded by:
\begin{align}
\overline{F}_Z(z)&=P[Z \geq z] \nonumber \\
&\leq 1-(1-e^{-a z})^T \nonumber \\
&=\sum_{i=1}^T \binom{T}{i}(-1)^{i+1}e^{-a zi}.
\end{align}

Therefore, the coverage probability in MF can be upper bounded by:
\begin{align}
P_{\text{cov,mf}}^{k}(K)&=\mathbb{E}_{r_k,I_r}\left[P\left[g_{k,\text{mf}} \geq \gamma r_k^{\alpha} I_r\right]|r_k,I_r\right] \nonumber \\
&\leq \sum_{l=1}^{L}\binom{L}{l}(-1)^{l+1} \mathbb{E}_{r_k,I_r}\left[e^{-\eta \gamma r_k^{\alpha} I_r l}|r_k,I_r\right] \nonumber \\
&=\sum_{l=1}^{L}\binom{L}{l}(-1)^{l+1} \mathbb{E}_{r_k}\left[\mathcal{L}_{I_r}(\eta \gamma r_k^{\alpha}l)|r_k\right]. \label{eqn:22}
\end{align}

In order to simplify the expression (\ref{eqn:22}), we make the mathematical transforms of $\mathcal{L}_{I_r}(s)$:
\begin{align}
\mathcal{L}_{I_r}(s)&=\left(\int_0^{r_k} \frac{1}{1+sr^{-\alpha}}\frac{2r}{r_k^2}dr\right)^{k-1} \nonumber \\
&~~~\times\exp\left(-2\pi\lambda_b\int_{r_k}^{\infty}\frac{sr^{-\alpha}}{1+sr^{-\alpha}}rdr\right) \nonumber \\
&=\left[1-\frac{2 s^{2/\alpha}}{\alpha r_k^2 }B\left(\frac{2}{\alpha},1-\frac{2}{\alpha},\frac{1}{1+sr_k^{-\alpha}}\right)\right]^{k-1} \nonumber \\
&~~~\times \exp\left[-2\pi\lambda_b \frac{s^{\frac{2}{\alpha}}}{\alpha}B^{'}\left(\frac{2}{\alpha},1-\frac{2}{\alpha},\frac{1}{1+sr_k^{-\alpha}}\right)\right],
\end{align}
where the last step follows by first replacing $s^{-\frac{1}{\alpha}}r$ with $u$, then replacing $\frac{1}{1+u^{-\alpha}}$ with $v$.

Therefore, the Laplace transform in (\ref{eqn:22}) can be written as:
\begin{align}
&\mathcal{L}_{I_r}(\eta \gamma r_k^{\alpha}l)=\left[1-\frac{2 (\eta \gamma l)^{2/\alpha}}{ \alpha}B\left(\frac{2}{\alpha},1-\frac{2}{\alpha},\frac{1}{1+\eta \gamma l}\right)\right]^{k-1} \nonumber\\
&~~~~~\times \exp\left[-2\pi\lambda_b \frac{r_k^2(\eta \gamma l)^{\frac{2}{\alpha}}}{\alpha} B^{'}\left(\frac{2}{\alpha},1-\frac{2}{\alpha},\frac{1}{1+\eta \gamma l}\right)\right] \nonumber \\
&~~~~~~~~~~~~~=\beta_1(\eta,\gamma,\alpha,l,k) \exp\left(-\pi\lambda_b r_k^2 \beta_2\left(\eta,\gamma,\alpha,l\right)\right), \label{eqn:laplace1}
\end{align}
where $\beta_1(\eta,\gamma,\alpha,l,k)$ and $\beta_2(\eta,\gamma,\alpha,l)$ are defined for notation simplicity, given by (\ref{eqn:beta1}) and (\ref{eqn:beta2}), respectively with $x=\eta$.

Then we need to calculate the expectation of $\mathcal{L}_{I_r}(\eta \gamma r_k^{\alpha}l)$ over $r_k$. It is observed that $\beta_1(\eta,\gamma,\alpha,l,k)$ can be seen as a constant and we only need to evaluate the expectation of $\exp(-\pi\lambda_b r_k^2 \beta_2(\eta,\gamma,\alpha,l))$ over $r_k$.
\begin{align}
&\mathbb{E}_{r_k}\left[\exp\left(-\pi\lambda_b r_k^2 \beta_2(\eta,\gamma,\alpha,l)\right)\right] \nonumber\\
&=\int_0^\infty \exp(-\pi\lambda_b r_k^2 \beta_2(\eta,\gamma,\alpha,l)) \frac{2(\lambda_b\pi r_k^2)^k}{r_k\Gamma(k)}\exp(-\lambda_b \pi r_k^2) dr_k \nonumber \\
&\overset{(a)}{=}\int_0^\infty \left[\frac{z}{1+\beta_2(\eta,\gamma,\alpha,l)}\right]^{k-1}\times\frac{e^{-z}}{\Gamma(k)\left(1+\beta_2(\eta,\gamma,\alpha,l)\right)}dz \nonumber \\
&\overset{(b)}{=}\left[\frac{1}{1+\beta_2(\eta,\gamma,\alpha,l)}\right]^k, \label{eqn:expectation}
\end{align}
where step (a) follows from the change of variables
\begin{equation}
z=\pi \lambda_b r_k^2 (1+\beta_2(\eta,\gamma,\alpha,l))  \label{eqn:change},
\end{equation}
and step (b) follows from the Gamma distribution property
\begin{equation}
\int_0^\infty t^k e^{-\lambda t}dt=\frac{k!}{\lambda^{k+1}}.
\end{equation}

By substituting (\ref{eqn:expectation}) into (\ref{eqn:22}), we obtain the upper bound of the coverage probability (\ref{eqn:app1}). The lower bound (\ref{eqn:app11}) can be similarly proved by letting $\eta=1$ in the above derivations. This theorem is thus proved.
\section{Proof of Theorem $2$}
With the similar methods used in the proof of Theorem 1, the coverage probability in ZF can be upper bounded by:
\begin{align}
P_{\text{cov,zf}}^{k}(K) &\leq \sum_{l=1}^{L-K+1}\binom{L-K+1}{l}(-1)^{l+1}    \nonumber\\
&~~~\times \mathbb{E}_{r_k,r_K}[\mathcal{L}_{I_r}(\kappa \gamma r_k^{\alpha}l)|r_k,r_K], \label{eqn:upper bound}
\end{align}
where $I_r$ is the inter-cluster interference from the SBSs farther than $\textbf{d}_k$. $\mathcal{L}_{I_r}(s)$ can be obtained by substituting $r_K$ for $r_k$ in (\ref{eqn:L2}) and is given by:
\begin{align}
\mathcal{L}_{I_r}(s)&=\exp\left[-2\pi\lambda_b\int_{r_K}^{\infty}\frac{sr^{-\alpha}}{1+sr^{-\alpha}}rdr\right].
\end{align}

By introducing a geometric parameter $\delta_k=\frac{r_k}{r_K}$, the Laplace transform in (\ref{eqn:upper bound}) can be written as:
\begin{align}
\mathcal{L}_{I_r}(\kappa \gamma r_k^{\alpha}l)&=\exp\left(-2\pi\lambda_b\int_{r_K}^\infty\frac{r^{-\alpha}\kappa \gamma r_k^{\alpha}l}{1+r^{-\alpha}\kappa \gamma r_k^{\alpha}l}rdr\right) \nonumber \\
&=\exp\left(-2\pi\lambda_b\int_{r_K}^\infty\frac{r}{1+(\frac{r}{r_K})^{\alpha}(\kappa \gamma \delta_k^{\alpha}l)^{-1}}dr\right) \nonumber \\
&=\exp\left(-\pi\lambda_b r_K^2(\kappa \gamma \delta_k^{\alpha}l)^{\frac{2}{\alpha}}\int_{(\kappa \gamma \delta_k^{\alpha}l)^{-\frac{2}{\alpha}}}^{\infty}\frac{1}{1+v^{\frac{\alpha}{2}}}dv\right),\label{eqn:Laplace}
\end{align}
where the last step follows from the change of variables
\begin{align}
v=\left[\left(\frac{1}{\kappa \gamma \delta_k^{\alpha}l}\right)^{\frac{1}{\alpha}}\frac{r}{r_K}\right]^2.
\end{align}

For notation simplicity, we let
\begin{align}
\beta_3(\kappa \gamma \delta_k^{\alpha}l,\alpha)=(\kappa \gamma \delta_k^{\alpha}l)^{\frac{2}{\alpha}}\int_{(\kappa \gamma \delta_k^{\alpha}l)^{-\frac{2}{\alpha}}}^{\infty}\frac{1}{1+v^{\frac{\alpha}{2}}}dv.  \label{eqn:beta3}
\end{align}

From (\ref{eqn:Laplace}), it is observed that we need to calculate the expectation over $\delta_k$ and $r_K$, rather than $r_k$ and $r_K$ as in (\ref{eqn:upper bound}). Thus, we first calculate the expectation of (\ref{eqn:Laplace}) over $r_K$.
\begin{align}
&\mathbb{E}_{r_K}\left[\mathcal{L}_{I_r}\left(\kappa \gamma (\delta_k r_K)^{\alpha}l\right)|\delta_k,r_K\right] \nonumber \\
&~~~~=\int_0^\infty \exp(-\pi\lambda_b r_K^2 \beta_3(\kappa \gamma \delta_k^{\alpha}l,\alpha)) \nonumber \\
&~~~~~~~\times \frac{2(\lambda_b\pi r_K^2)^K}{r_K\Gamma(K)}\exp(-\lambda_b \pi r_K^2) dr_K  \nonumber\\
&~~~~=\frac{1}{\left[1+\beta_3(\kappa \gamma \delta_k^{\alpha}l,\alpha)\right]^K}, \label{eqn:beta33}
\end{align}
where the last step follows from change of variables similar to (\ref{eqn:change}).

Therefore, $P_{\text{cov,zf}}^{k,\text{u}}(K)$ can be rewritten as:
\begin{align}
P_{\text{cov,zf}}^{k,\text{u}}(K)=\mathbb{E}_{\delta_k}\left[\sum_{l=1}^{L-K+1} \frac{\binom{L-K+1}{l}(-1)^{l+1}}{\left[1+\beta_3(\kappa \gamma \delta_k^{\alpha}l,\alpha)\right]^K}\right]. \label{eqn:upper bound1}
\end{align}

To obtain the expectation above over $\delta_k$, we need to know the pdf of $\delta_k$. Utilizing the joint pdf of $r_k$ and $r_K$ given in (\ref{eqn:joint pdf}), the CDF of $\delta_k$ is given by:
\begin{align}
P[\delta_k \leq x]&=P[r_k \leq xr_K]   \nonumber\\
&=\int_0^\infty \int_0^{xr_K} f_{R_k,R_K}(r_k,r_K)dr_kdr_K \nonumber\\
&=\int_0^\infty \int_0^{xr_K}\frac{4r_k r_k^{2(k-1)}r_K}{\Gamma(K-k)\Gamma(k)}(\lambda_b\pi)^K  \nonumber\\
&~~~\times (r_K^2-r_k^2)^{K-k-1}\exp(-\lambda_b \pi r_K^2)dr_kdr_K \nonumber\\
&=1-\sum_{i=0}^{k-1}\frac{(K-1)!x^{2(k-1-i)}(1-x^2)^{K-k+i}}{(K-k+i)!(k-1-i)!}, \label{eqn:CDF}
\end{align}
where $0\leq x \leq 1$. Then, the pdf of $\delta_k$ can be obtained as
\begin{align}
f_{\delta_k}(x)&=\frac{dP[\delta_k \leq x]}{dx} \nonumber\\
&=\sum_{i=0}^{k-1} \frac{(K-1)!\left[(K-1)x^2-(k-i-1)\right]}{(K-k+i)!(k-1-i)!} \nonumber\\
&~~~\times 2x^{2(k-1-i)-1}(1-x^2)^{K-k+i-1} \nonumber\\
&=\frac{2(K-1)!}{(k-1)!(K-k-1)!}x^{2k-1}(1-x^2)^{K-k-1}.
\end{align}

Recall (\ref{eqn:beta3}), we approximate the integral in it as a constant value according to the randomness of $\delta_k$, which is given by:
\begin{align}
\mathbb{E}\left[\int_{\delta_k^{-2}(\kappa \gamma l)^{-\frac{2}{\alpha}}}^{\infty}\frac{1}{1+v^{\frac{\alpha}{2}}}dv\right]\simeq \sqrt{\frac{k}{K}} \mathcal{A}\left(\frac{\sqrt K (\kappa \gamma l)^{-\frac{2}{\alpha}}}{\sqrt k }\right),
\end{align}
which follows from the expectation of $\delta_k^2$:
\begin{align}
\mathbb{E}(\delta_k^2)&=\int_0^1 x^2f_{\delta_k}(x)dx  \nonumber \\
&=\int_0^1\frac{2 x^2 (K-1)!}{(k-1)!(K-k-1)!}x^{2k-1}(1-x^2)^{K-k-1}dx \nonumber \\
&=\frac{k}{K}.
\end{align}

Thus, we can approximate $\beta_3(\kappa \gamma \delta_k^{\alpha}l,\alpha)$ as:
\begin{align}
\beta_3(\kappa \gamma \delta_k^{\alpha}l,\alpha)\simeq \delta_k^2(\kappa \gamma l)^{\frac{2}{\alpha}}\sqrt{\frac{k}{K}} \mathcal{A}\left(\frac{\sqrt K (\kappa \gamma l)^{-\frac{2}{\alpha}}}{\sqrt k }\right). \label{eqn:appbeta3}
\end{align}

By substituting (\ref{eqn:appbeta3}) into (\ref{eqn:beta33}), the expectation of (\ref{eqn:beta33}) over $\delta_k$ can be approximated as:
\begin{align}
&\mathbb{E}_{\delta_k}\left[\left(\frac{1}{1+\beta_3(\kappa \gamma \delta_k^{\alpha}l,\alpha)}\right)^K\right] \nonumber \\
&=\int_0^1\left[\frac{1}{1+\beta_3(\kappa \gamma x^{\alpha}l,\alpha)}\right]^K f_{\delta_k}(x)dx \nonumber \\
&\simeq\int_0^1\frac{f_{\delta_k}(x)}{\left[1+(\kappa \gamma l)^{\frac{2}{\alpha}}\sqrt{\frac{k}{K}} \mathcal{A}\left(\frac{\sqrt K (\kappa \gamma l)^{-\frac{2}{\alpha}}}{\sqrt k }\right)x^2\right]^K} dx \nonumber \\
&=\frac{1}{{\left[1+(\kappa \gamma l)^{\frac{2}{\alpha}}\sqrt{\frac{k}{K}} \mathcal{A}\left(\frac{\sqrt K (\kappa \gamma l)^{-\frac{2}{\alpha}}}{\sqrt k }\right)\right]^k}}.  \label{eqn:Edelta}
\end{align}

Substituting (\ref{eqn:Edelta}) into (\ref{eqn:upper bound1}), we obtain the approximate upper bound of the coverage probability (\ref{eqn:app2}). The approximate lower bound (\ref{eqn:app22}) can be similarly proved by letting $\kappa=1$ in the above derivations. This theorem is thus proved.

\section{Proof of Lemma $5$}
Before proving that the optimization problem $\mathbf{P1}$ is convex, we first need to prove that the coverage probability $P_{\text{cov}}^k(K)$ is a non-increasing function with respect to $k$. This result is intuitive because the average SIR for the user served by the farther SBS is lower, which causes the smaller coverage probability. Next, we prove it mathematically.

In the ZF case, for ease of notation, we let
\begin{align}
X(s)=-2\pi\lambda_b\int_{r_K}^{\infty}\frac{sr^{-\alpha}}{1+sr^{-\alpha}}rdr.
\end{align}

Hence, the Laplace transform of the interference $I_r$ is $\mathcal{L}_{I_r}(s)=\exp\left[X(s)\right]$ and its $i$-th order derivative is given by:
\begin{align}
\mathcal{L}_{I_r}^{(i)}(s)=\sum_{m=0}^{i-1}\binom{i-1}{m}\mathcal{L}_{I_r}^{(m)}(s)X^{(i-m)}(s).
\end{align}

It is observed that the even order derivative of $X(s)$ is positive and its odd order derivative is negative, which can be proved easily by mathematical induction. Then we define a function $G(s)=\sum_{i=0}^{L-K}\frac{(-s)^i}{i!}\mathcal{L}_{I_r}^{(i)}(s)$ and its derivative with respect to s is given by:
%

\begin{align}
G^{'}(s)&=\sum_{i=0}^{L-K}\frac{(-1)^i}{i!}[is^{i-1}\mathcal{L}_{I_r}^{(i)}(s)+s^i\mathcal{L}_{I_r}^{(i+1)}(s)] \nonumber\\
&=(-1)^{L-K}s^{L-K}\mathcal{L}_{I_r}^{(L-K+1)}(s)  \nonumber\\
&\leq0,
\end{align}
which means $G(s)$ is non-increasing function with respect to $s$. Therefore, $G(\gamma r_k^{\alpha})$ is an non-increasing function with respect to $r_k$. Thus, for a fixed topology of SBSs in the plane, we always have $G(\gamma r_k^{\alpha})\geq G(\gamma r_{k+1}^{\alpha})$. Since the coverage probability $P_{\text{cov,zf}}^{k}(K)=\mathbb{E}_{r_k,r_K}[G(\gamma r_k^{\alpha})]$, (\ref{eqn:cov-zf}) is proven to be a non-increasing function with respect to $k$. For the MF case, the proof is similar, and hence is omitted here.


Then we prove the non-increasing property of the approximate coverage probability (\ref{eqn:app1}) and (\ref{eqn:app2}) with respect to $k$.

In the ZF case, we define a non-negative random variable $I_3$ similar to the interference $I_r$ and its Laplace transform is given by:
\begin{align}
\mathcal{L}_{I_3}(s)&=\exp\left(-2\pi\lambda_b\int_{r_k}^{\infty}\frac{sr^{-\alpha}}{1+sr^{-\alpha}}rdr\right).
\end{align}

Similar to (\ref{eqn:Laplace}), we have:
\begin{align}
&\mathcal{L}_{I_3}\left(\kappa \gamma (\delta_kr_K)^{\alpha}l\left(\frac{k}{K}\right)^{\frac{\alpha}{4}}\right)  \nonumber\\
&~~=\exp\left(-\pi\lambda_b r_K^2\delta_k^2(\kappa \gamma l)^{\frac{2}{\alpha}}\sqrt{\frac{k}{K}} \mathcal{A}\left(\frac{\sqrt K (\kappa \gamma l)^{-\frac{2}{\alpha}}}{\sqrt k }\right)\right) \nonumber\\
&~~=\exp\left(-\pi\lambda_b r_k^2(\kappa \gamma l)^{\frac{2}{\alpha}}\sqrt{\frac{k}{K}} \mathcal{A}\left(\frac{\sqrt K (\kappa \gamma l)^{-\frac{2}{\alpha}}}{\sqrt k }\right)\right),
\end{align}

Therefore, the approximate coverage probability in ZF case can be written as:
\begin{align}
P_{\text{cov,zf}}^{k,\text{u}}(K)&=\sum_{l=1}^{L-K+1}\binom{L-K+1}{l}(-1)^{l+1}    \nonumber\\
&~~~\times \mathbb{E}_{r_k}\left[\mathcal{L}_{I_3}\left(\kappa \gamma r_k^{\alpha}l\left(\frac{k}{K}\right)^{\frac{\alpha}{4}}\right)\bigg|r_k\right] \nonumber\\
&=\mathbb{E}_{r_k,I_3}\bigg[1-\bigg[1-\exp\bigg(-I_3\kappa \gamma r_k^{\alpha}    \nonumber\\
&~~~~~~~~~~~~~~~~~~~~\times\left(\frac{k}{K}\right)^{\frac{\alpha}{4}}\bigg)\bigg]^{L-K+1}\bigg|r_k,I_3\bigg],
\end{align}
which is a non-increasing function with respect to $k$. Hence, we conclude that:
\begin{align}
P_{\text{cov,zf}}^{k,\text{u}}(K)\geq\sum_{l=1}^{L-K+1} \frac{\binom{L-K+1}{l}(-1)^{l+1}}{{\left[1+\left(\kappa \gamma l\right)^{\frac{2}{\alpha}}\sqrt{\frac{k+1}{K}} \mathcal{A}\left(\frac{\sqrt K \left(\kappa \gamma l\right)^{-\frac{2}{\alpha}}}{\sqrt{k+1} }\right)\right]^k}}. \label{eqn:uneuqality1}
\end{align}

Then, we define a non-negative random variable $I_4$ similarly and its Laplace transform is given by:
\begin{align}
\mathcal{L}_{I_4}(s)&=\mathcal{L}_{I_3}(s)\times\int_{r_k}^\infty \frac{sr^{-\alpha}}{1+sr^{-\alpha}}\frac{2r}{r_k^2}dr   \nonumber \\
&~~/\left(1+\int_{r_k}^\infty \frac{sr^{-\alpha}}{1+sr^{-\alpha}}\frac{2r}{r_k^2}dr\right).
\end{align}

Hence, we have
\begin{align}
&\sum_{l=1}^{L-K+1} \frac{\binom{L-K+1}{l}(-1)^{l+1}}{{\left[1+\left(\kappa \gamma l\right)^{\frac{2}{\alpha}}\sqrt{\frac{k+1}{K}} \mathcal{A}\left(\frac{\sqrt K \left(\kappa \gamma l\right)^{-\frac{2}{\alpha}}}{\sqrt{k+1} }\right)\right]^k}}-P_{\text{cov,zf}}^{k+1,\text{u}}(K) \nonumber\\
&~~=\mathbb{E}_{r_k,I_4}\left[1-\left(1-e^{-I_{4}\kappa \gamma r_k^{\alpha}\left(\frac{k+1}{K}\right)^{\frac{\alpha}{4}}}\right)^{L-K+1} \right] \nonumber\\
&~~\geq 0. \label{eqn:uneuqality2}
\end{align}

Combining (\ref{eqn:uneuqality1}) and (\ref{eqn:uneuqality2}), we conclude that $P_{\text{cov,zf}}^{k,\text{u}}(K)\geq P_{\text{cov,zf}}^{k+1,\text{u}}(K)$, which means the approximate coverage probability in ZF is a non-increasing function with respect to $k$.

In the MF case, we define a random variable $I_5$ similarly and its Laplace transform is given by:
\begin{align}
\mathcal{L}_{I_5}(s)&=\left(\int_0^{r_k} \frac{1}{1+sr^{-\alpha}}\frac{2r}{r_k^2}dr\right)^{k-1} \times \mathcal{L}_{I_3}(s) \nonumber \\
&~~\times\int_0^\infty \frac{sr^{-\alpha}}{1+sr^{-\alpha}}\frac{2r}{r_k^2}dr/\left(1+\int_{r_k}^\infty \frac{sr^{-\alpha}}{1+sr^{-\alpha}}\frac{2r}{r_k^2}dr\right).
\end{align}

Hence, we have
\begin{align}
&P_{\text{cov,mf}}^{k,\text{u}}(K)-P_{\text{cov,mf}}^{k+1,\text{u}}(K)  \nonumber\\
&~~=\sum_{l=1}^{L} \beta_1(\eta,\gamma,\alpha,l,k)\frac{\binom{L}{l}(-1)^{l+1}}{\left[1+\beta_2(\eta,\gamma,\alpha,l)\right]^k} \nonumber\\
&~~~~~~\times\frac{1-\beta_1(\eta,\gamma,\alpha,l,2)+\beta_2(\eta,\gamma,\alpha,l)}{1+\beta_2(\eta,\gamma,\alpha,l)} \nonumber\\
&~~=\mathbb{E}_{r_k,I_5}\left[1-\left(1-e^{-I_{5}\eta \gamma r_k^{\alpha}}\right)^{L} \right] \nonumber\\
&~~\geq 0,
\end{align}
which means the approximate coverage probability in MF $P_{\text{cov,mf}}^{k,\text{u}}(K)$ is a non-increasing function with respect to $k$.

Therefore, we conclude that the exact and approximate coverage probability $P_{\text{cov}}^k(K)$ is a non-increasing function with respect to $k$ in both ZF and MF schemes.

Utilizing the non-increasing property of $P_{\text{cov}}^k(K)$ with respect to $k$ proved above, the second order derivative of the objective function STP with respect to $b_n$ is given by:
\begin{align}
&\frac{\partial^2 P_{\text{suc}}(K)}{\partial b_n^2} \nonumber\\
&=\sum_{n=1}^N p_n{\sum_{k=1}^K(k-1)(1-b_n)^{k-3}(kb_n-2)P_{\text{cov}}^k(K)} \nonumber\\
&=\sum_{n=1}^N p_n\Bigg[-2P_{\text{cov}}^2(K)+2(3b_n-2)P_{\text{cov}}^3(K) \nonumber\\
&~~~+\sum_{k=4}^K(k-1)(1-b_n)^{k-3}(kb_n-2)P_{\text{cov}}^k(K)\Bigg] \nonumber\\
&\leq \sum_{n=1}^N p_n\Bigg[6(b_n-1)P_{\text{cov}}^3(K)+3(1-b_n)(4b_n-2)P_{\text{cov}}^4(K) \nonumber\\
&~~~+\sum_{k=5}^K(k-1)(1-b_n)^{k-3}(kb_n-2)P_{\text{cov}}^k(K)\Bigg] \nonumber\\
&\cdots\cdots   \nonumber\\
&\leq K(K-1)(1-b_n)^{K-3}(b_n-1)P_{\text{cov}}^K(K) \nonumber\\
&\leq 0,
\end{align}
where the first $K-2$ inequalities follow from the non-increasing property of the coverage probability $P_{\text{cov}}^k(K)$ with respect to $k$. The last step follows from that $0\leq b_n\leq 1$. Thus, the objective function (\ref{eqn:objective}) we want to maximize is a concave function, and both constraints (\ref{eqn:contraint1}) and (\ref{eqn:contraint2}) are linear. Therefore, the proof is completed.
\section{Proof of Theorem $3$}
The Lagrangian function is given by:
\begin{align}
L(b_1,b_2,\cdots,b_N,\mu)&=\sum_{n=1}^N p_n{\sum_{k=1}^K b_n(1-b_n)^{k-1}P_{\text{cov}}^{k}(K)} \nonumber\\
&~~~+\mu\left(M-\sum_{n=1}^Nb_n\right),
\end{align}
where $u$ is the Lagrangian multiplier associated with the constraint (\ref{eqn:constraint}). The partial derivative of the Lagrangian function with respect to $b_n$ is given by:
\begin{align}
{\cal L}&=\frac{\partial L(b_1,b_2,\cdots,b_N,\mu)}{\partial b_n} \nonumber\\
&=p_n \sum_{k=1}^K(1-b_n)^{k-2}(1-kb_n)P_{\text{cov}}^k(K)-\mu.
\end{align}

By letting ${\cal L}=0$, we have
\begin{align}
p_n \sum_{k=1}^K(1-b_n)^{k-2}(1-kb_n)P_{\text{cov}}^k(K)=\mu. \label{eqn: Lagrangian}
\end{align}

It is observed that the left hand of (\ref{eqn: Lagrangian}) is a decreasing function with respect to $b_n$ since the objective function is concave. Notice that $0\leq b_n \leq 1$. Thus, when $b_n=1$, $\mu$ has the minimum value: $p_n\left[P_{\text{cov}}^1(K)-P_{\text{cov}}^2(K)\right]$. While for $b_n=0$, it has the maximum value: $p_n\sum_{k=1}^KP_{\text{cov}}^k(K)$.
Therefore, for a given Lagrangian multiplier $\mu$, the optimal cache solutions $b_n(\mu)$ are given by:
\begin{align}
b_n(\mu)=
\begin{cases}
1,&\mu \leq p_n\left[P_{\text{cov}}^1(K)-P_{\text{cov}}^2(K)\right] \\
w_n(\mu),&\text{otherwise}\\
0,&\mu \geq p_n\sum_{k=1}^KP_{\text{cov}}^k(K)
\end{cases}, \label{eqn:strategy}
\end{align}
which is equivalent to (\ref{eqn:11111}) by substituting $\mu^*$ for $\mu$ in (\ref{eqn:strategy}), which is the optimal dual variable satisfying the cache size constraint (\ref{eqn:constraint}). Hence, the proof is completed.
\end{appendices}
\bibliographystyle{IEEEtran}
\bibliography{IEEEabrv,paper1}
\end{document}